\documentclass[copyright,creativecommons]{eptcs}

\providecommand{\event}{DCM 2013}

\usepackage{multicol}
\usepackage{needspace}
\usepackage{float}
\usepackage[rflt]{floatflt}
\floatstyle{boxed}\restylefloat{table}
\hypersetup{
  colorlinks=true,
  citecolor=blue,
  linkcolor=red,
  urlcolor=blue,
  pdftitle={The probability of non-confluent systems},
  pdfauthor={Alejandro D\'iaz-Caro and Gilles Dowek}
}

\usepackage[all]{xy}
\usepackage[only,llbracket,rrbracket,llparenthesis,rrparenthesis]{stmaryrd}
\usepackage{proof}
\usepackage{wrapfig}
\usepackage{cite}
\usepackage{amsmath,amssymb,amsthm}
\usepackage{amsfonts}

\DeclareSymbolFont{cmmathcal}{OMS}{cmsy}{m}{n}
\DeclareSymbolFontAlphabet{\mathcal}{cmmathcal}
\newtheorem{theorem}{Theorem}[section]
\newtheorem{lemma}[theorem]{Lemma}
\newtheorem{corollary}[theorem]{Corollary}
\theoremstyle{definition}
\newtheorem{definition}[theorem]{Definition}
\newtheorem{example}[theorem]{Example}
\theoremstyle{remark}
\newtheorem{remark}[theorem]{Remark}

\newcommand{\sui}[1]{\ensuremath{\sum_{i=1}^{#1}}}
\newcommand{\suj}[1]{\ensuremath{\sum_{j=1}^{#1}}}

\newcommand{\ve}[1]{\ensuremath{\mathbf{#1}}}
\newcommand{\parts}[1]{\ensuremath{\mathcal{P}({#1})}}
\newcommand{\A}{\ensuremath{\mathcal{A}}}
\newcommand{\p}{\ensuremath{\mathtt{p}}}
\newcommand{\Pb}{\ensuremath{\mathtt{P}}}
\newcommand{\no}[1]{\ensuremath{#1^\sim}}
\newcommand{\Alg}{\ensuremath{\mbox{Alg}_{F}^{p}}}
\newcommand{\Lplus}{\ensuremath{\lambda_+}}
\newcommand{\Lpplus}{\ensuremath{\lambda_+^p}}

\newcommand{\arriba}[2]{\ensuremath{\genfrac{}{}{0pt}{}{#1}{#2}}}
\newcommand{\notder}{\ensuremath{{\not\phantom{\ }\!\!\!\!:}~}}
\newcommand{\AtP}[1]{\ensuremath{\llbracket{#1}\rrbracket}}
\newcommand{\PtA}[1]{\ensuremath{\llparenthesis{#1}\rrparenthesis}}
\newcommand{\eq}{\ensuremath{\rightleftarrows}}
\newcommand{\re}{\ensuremath{\hookrightarrow}}

\newcommand{\cond}[1]{\ensuremath{\scriptstyle[#1]\,}}
\newcommand{\condi}[2]{\ensuremath{{#1}{\vcenter{#2}}}}
\newcommand{\coherentContext}[1]{\ensuremath{{#1}^f}}

\title{The probability of non-confluent systems}
\author{Alejandro D\'iaz-Caro
  \institute{
    \parbox{3.8cm}{
      \centering\vspace{1.3mm}
      Universit\'e Paris Ouest\\[-0.5mm]
      200 avenue de la R\'epublique\\[-0.5mm]
      92001 Nanterre, France
    }\qquad
    \parbox{3.8cm}{
      \centering\vspace{1.3mm}
      INRIA\\[-0.5mm]
      23 avenue d'Italie, CS 81321\\[-0.5mm]
      75214 Paris Cedex 13, France
    }\\[4mm]
  }\email{alejandro@diaz-caro.info}
  \and
  Gilles Dowek
  \institute{
    INRIA\\
    23 avenue d'Italie, CS 81321\\
    75214 Paris Cedex 13, France
  }\email{gilles.dowek@inria.fr}
}

\def\titlerunning{The probability of non-confluent systems}
\def\authorrunning{A. D\'iaz-Caro \& G. Dowek}

\begin{document}
\maketitle

\begin{abstract}
  We show how to provide a structure of probability space to the set of execution traces on a non-confluent abstract rewrite system, by defining a variant of a Lebesgue measure on the space of traces. Then, we show how to use this probability space to transform a non-deterministic calculus into a probabilistic one. We use as example $\Lplus$, a recently introduced calculus defined through type isomorphisms.
\end{abstract}

\section{Introduction}\label{sec:intro}
Many probabilistic calculi has been developed in the pasts years, e.g. \cite{
  BournezHoyrupRTA03,
  HerescuPalamidessiFOSSACS00,
  DipierroHankinWiklickyJLC05,
  DalLagoZorziRAIRO12,
  AndovaFMRTPS99
}. In particular, the algebraic versions of $\lambda$-calculus~\cite{ArrighiDowekRTA08,VauxMSCS09} are extensions to $\lambda$-calculus where a linear combination of terms, e.g.~$\alpha.\ve r+\beta.\ve s$, is also a term. One way to interpret such a linear combination is that it represents a term which is the term $\ve r$ with probability $\alpha$, or the term $\ve s$ with probability $\beta$. However, endowing such a calculus with a non-restrictive type system is a challenge~\cite{ArrighiDiazcaroValironDCM11,ArrighiDiazcaroValiron13}.

A simpler framework is that of non determinisitic calculi which can be seen as algebraic calculi withouth scalars. They have been studied, for instance in~\cite{
  BoudolIC94,
  BucciarelliEhrhardManzonettoAPAL12,
  deLiguoroPipernoIC95,
  DezaniciancagliniDeliguoroPipernoSIAM98,
  DiazcaroDowekLSFA12,
  DiazcaroDowek13,
  DiazcaroManzonettoPaganiLFCS13,
  DiazcaroPetitWoLLIC12,
PaganiRonchidellaroccaFI10}, however moving back from non-determinism to probabilities is not trivial.
In this paper we propose, instead of changing these models, to define a probability measure on reductions in non-deterministic systems.
In fact, as we shall see, such a probability measure can be defined on any abstract non-deterministic transition systems, or non-confluent abstract rewrite systems (ARS)~(cf.~\cite[Chapter~1]{Terese03}).
Our goal is to show that explicit probabilities are not needed in the syntax, and that the simpler non-deterministic calculi are as powerful as the more complicated probabilistic calculi.

\begin{wrapfigure}{r}{0.2\textwidth}
  \begin{center}\vspace{-0.3in}
    $\xymatrix@C=2ex@R=4ex{
      & \ve a \ar[dl] \ar[dr] & & \\
      \ve b & 					  & \ve c \ar[dl]\ar[dr] &\\
      & \ve d				  & & \ve e
    }$\end{center}
\end{wrapfigure} Consider for example the following non-confluent ARS
\begin{displaymath}
  \ve a\to\ve b{\enspace,\qquad}
  \ve a\to\ve c{\enspace,\qquad}
  \ve c\to\ve d{\enspace,\qquad}
  \ve c\to\ve e{\enspace,\qquad}
\end{displaymath}
we want to associate a probability to events such as
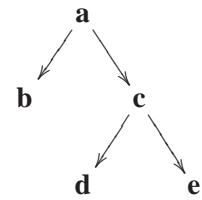
\begin{displaymath}
  \ve a\to^*\ve b{\enspace,\qquad}
  \ve a\to^*\ve c{\enspace,\qquad}
  \ve a\to^*\ve d{\enspace,\qquad}
  \ve a\to^*\ve e{\enspace.\qquad}
\end{displaymath}
In this example, assuming equiprobability, we have $\Pb(\ve a\to^*\ve b) = \tfrac{1}{2}$, $\Pb(\ve a\to^*\ve c) = \tfrac{1}{2}$, $\Pb(\ve a\to^*\ve d) = \tfrac{1}{4}$, $\Pb(\ve a\to^*\ve e) = \tfrac{1}{4}$. Notice that these events are not disjoints and that their sum is larger than $1$. In particular, $\ve a\to^*\ve d$ implies $\ve a\to^*\ve c$.
Defining the elements of the set $\Omega$ of elementary events is not completely straightforward, in particular because we want to make it general enough to also consider infinite cases. For example, in the following system
\begin{displaymath}
  \ve a_i \to\ve a_{i+1},\quad \ve a_i \to\ve a'_{i+1}\enspace,
\end{displaymath}
we naturally would like that $\Pb(\ve a_0 \to^*\ve a_n) = \tfrac{1}{2^n}$.

Besides defining the elements of the set $\Omega$, we need to define a notion of a measurable subset of $\Omega$ and endow such a subset with a probability distribution verifying the Kolmogorov axioms.

Our idea is to follow Lebesgue: define first the probability of rectangles, or boxes, then the probability of any set and finally measurable sets as those that verify Lebesgue's property.
Thus besides defining the set $\Omega$, we need to define a subset of $\parts{\Omega}$ of boxes.

The first intuition would be to take paths as elements of the set $\Omega$, for instance assigning the probability $\tfrac{1}{2}$ to the paths $\ve a\to\ve b$, $\tfrac{1}{4}$ to $\ve a\to\ve c\to\ve d$ and $\tfrac{1}{4}$ to $\ve a\to\ve c\to\ve e$. In fact it seems more convenient to extend such paths to strategies prescribing one reduct for each non normal object. Boxes are then defined as sets of strategies agreeing on a finite domain. We show in this paper that this is sufficient to define a probability space on strategies, consistent with the intuitive probability of events of the form $\ve a\to^*\ve b$.

Our study is generic enough to be applicable to several settings, such as automatons, or any other kind of transition systems. We use the nomenclature of abstract rewrite systems, but that of states and transitions could be used as well. Finally, we apply this construction to \Lplus~\cite{DiazcaroDowekLSFA12,DiazcaroDowek13}.

\paragraph{Plan of the paper.}
Section~\ref{sec:prelim} introduces the basic concepts of strategies and boxes, it defines the Lebesgue measures. Section~\ref{sec:prob} proves that the space of strategies forms a probability space.
Finally, in section~\ref{sec:Lpplus} we show how to modify the calculus \Lplus\ into a probabilistic calculus \Lpplus. Also, we provide an encoding of an algebraic $\lambda$-calculus into \Lpplus\ and, to some extend, the inverse translation.

\section{Preliminaries}\label{sec:prelim}

Let $\Lambda$ be a set of objects and $\to$ a function from $\Lambda\times\Lambda$ to $\mathbb{N}$ such that for all $\ve a$ the set $\{\ve b~|~\to(\ve a,\ve b)\neq 0\}$ is finite. 
We write $\ve a\to\ve b$ if $\to(\ve a,\ve b)\neq 0$. 
We allow a term to be written to the same symbol more than once, so its probability increases, e.g.\ if $\to(\ve a,\ve b)=2$ and $\to(\ve a,\ve c)=1$, then the probability of getting $\ve b$ will be the double than the probability of getting $\ve c$. Think for example in a non-deterministic choice between two objects, which happen to be equal, then there would be two ways to get such an object by doing the choice.
For a given object $\ve a\in\Lambda$, we denote by $\rho(\ve a)$ its degree, that is, the number of objects to which it can be rewritten to in one step. Definition~\ref{def:degree} formalises this.

\begin{definition}[Degree of an object]\label{def:degree}
  $\rho:\Lambda\to\mathbb{N}$ is a function defined by
  $
  \rho(\ve a)=\sum_{\ve b}\to(\ve a,\ve b)
  $.
\end{definition}

An object is normal if its degree is $0$.
We denote by $\Lambda^+=\{\ve a~|~\ve a\in\Lambda$ and $\rho(\ve a)\geq 1\}$ to the set of non-normal objects, that is, objects that can be rewritten to other objects.

A strategy prescribing one reduct for each non-normal object is defined as a function from $\Lambda^+$ to $\Lambda$ (cf.~\cite[Def.~4.9.1]{Terese03}).
\begin{definition}[Strategy]\label{def:strategy}
  A strategy is a total function $f:\Lambda^+\to\Lambda$ such that $f(\ve a)=\ve b$ implies $\ve a\to\ve b$. For instance, if $\ve a\to\ve b$ and $\ve a\to\ve b'$, there are two functions, $f$ and $f'$ assigning different results to $\ve a$.
  We denote by $\Omega$ the set of all such functions.
\end{definition}

A box is a set of strategies agreeing on a finite domain.

\begin{definition}[Box]\label{def:box}
  A box $B\subseteq\Omega$ is a set of the form $\{f~|~f(\ve a_1)=\ve a_1',\dots,f(\ve a_n)=\ve a_n'\}$ for some objects $\ve a_i$, $\ve a_i'$.
  We write $\mathcal{B}(\Omega)$ the subset of $\parts{\Omega}$ containing all the boxes.
\end{definition}

\begin{example}\label{ex:box}
  Continuing with the example given at the introduction, $\Lambda^+=\{\ve a,\ve c\}$. Let $f_1(\ve a)=\ve b, f_1(\ve c)=\ve d$ and $f_2(\ve a)=\ve b, f_2(\ve c)=\ve e$ be two of the four strategies of $\Omega$. Then the box $\{f~|~f(\ve a)=\ve b, f(\ve c)=\ve d\}=\{f_1\}$, and the box $\{f~|~f(\ve a)=\ve b\}$ is $\{f_1,f_2\}$.
  \[
    \left\{
      \begin{array}{c@{\quad=\quad}c@{\qquad;\qquad}c@{\quad=\quad}c}
	f_1
	&
	\parbox{4cm}{
	  \xymatrix@C=2ex@R=2ex{
	    & \ve a \ar[dl] & \\
	    \ve b & & \ve c \ar[dl]\\
	    & \ve d &
	  }
	}
	&
	f_2
	&
	\parbox{4cm}{
	  \xymatrix@C=2ex@R=2ex{
	    & \ve a \ar[dl] & & \\
	    \ve b & & \ve c \ar[dr] & \\
	    & & & \ve e
	  }
	}
      \end{array}
    \right\}
    \qquad=\qquad
    \parbox{2,5cm}{
      \vspace{-0.5cm}\centering Box\\
      \framebox[2cm]{
	\xymatrix@C=2ex@R=2ex{
	  & \ve a \ar[dl] \\
	  \ve b &
	}
      }
    }
  \]
\end{example}

A probability distribution can be defined in term of boxes, and then be extended to arbitrary sets of strategies.

\begin{definition}[Probability function]\label{def:prob}
  Let $\p:\mathcal{B}(\Omega)\to[0,1]$ be a total function defined over boxes as follows. If $B=\{f~|~f(\ve a_1)=\ve a_1',\dots,f(\ve a_n)=\ve a_n'\}$, then
  \begin{displaymath}
    \p(B)=\prod_{i=1}^n\frac{\to(\ve a_i,\ve a_i')}{\rho(\ve a_i)}\enspace.
  \end{displaymath}
  By convention, if no condition is given in $B$ (i.e.~$B=\Omega$), we have $n=0$, and we consider the product of zero elements to be $1$, the neutral element of the product.

  Then we define the probability measure $\Pb:\parts{\Omega}\to[0,1]$ for arbitrary sets of strategies as follows
  \begin{displaymath}
    \Pb(S)=\left\{
      \begin{array}{ll}
	0
	& \mbox{if }S=\emptyset\\
	\inf\left\{\sum_{B\in\mathcal{C}}\p(B)~|~\mathcal{C}\mbox{ is a countable family of boxes s.t. }S\subseteq\bigcup_{B\in\mathcal{C}}B\right\}
	& \mbox{in other case}
      \end{array}
    \right.%}
  \end{displaymath}
\end{definition}

\begin{example}\label{ex:probOfBox}
  Consider the ARS $\ve a\to\ve b$ with multiplicity $2$ and $\ve a\to\ve c$ with multiplicity $1$. 
  \begin{floatingfigure}[r]{0.16\textwidth}
    \vspace{-20pt}
    \begin{center}
      $
      \xymatrix@C=2ex@R=2ex{
	& \ve a \ar[d]\ar[dr]\ar[dl] & & \\
	\ve b & \ve b & \ve c
      }
      $
    \end{center}
  \end{floatingfigure}  Let $B$ be the box $B=\{f~|~f(\ve a)=\ve b\}$. Then we have $\p(B)=\frac{\to(\ve a,\ve b)}{\rho(\ve a)}=\frac{2}{3}$.
  Intuitively, $\Pb(B)$ is the same as $\p(B)$ (this will be later formalised in Lemma~\ref{lem:Peqp}), because $B$ is the minimum cover of $B$, that is, $\{B\}$ is the minimum family of boxes such that $B$ is in its union. 
  Hence $\Pb(B)=\frac{2}{3}$.
\end{example}
\begin{example}\label{ex:probOfS}
  We continue with the same running example depicted in the introduction. Let $f_1(\ve a)=\ve b$, $f_1(\ve c)=\ve d$ and $f_3(\ve a)=\ve c$, $f_3(\ve c)=\ve e$ be two strategies. Then the set $S=\{f_1,f_3\}$ is minimally covered by the boxes $B_1=\{f_1\}=\{f~|~f(\ve a)=\ve b, f(\ve c)=\ve d\}$ and $B_2=\{f_3\}=\{f~|~f(\ve a)=\ve c, f(\ve c)=\ve e\}$. So we have $P(S)=\p(B_1)+\p(B_2)=\frac{1}{2\times 2}+\frac{1}{2\times 2}=\frac{1}{2}$.
  \[
    S=\left\{
      \begin{array}{c@{\quad=\quad}c@{\qquad;\qquad}c@{\quad=\quad}c}
	f_1
	&
	\parbox{4cm}{
	  \xymatrix@C=2ex@R=2ex{
	    & \ve a \ar[dl] & \\
	    \ve b & & \ve c \ar[dl]\\
	    & \ve d &
	  }
	}
	&
	f_3
	&
	\parbox{4cm}{
	  \xymatrix@C=2ex@R=2ex{
	    \ve a \ar[dr] & & \\
	    & \ve c \ar[dr] & \\
	    & & \ve e
	  }
	}
      \end{array}
    \right\}
  \]
\end{example}

Now we can define the Lebesgue measure in terms of the given probability measure.
\begin{definition}[Measurable]\label{def:lebesgue}
  Let $A$ be an element of $\parts{\Omega}$, we write $\no{A}$ for the complement of $A$, that is $\Omega\setminus A$. The set $A$ is Lebesgue measurable if $\forall S\in\parts{\Omega}$, we have
  $\Pb(S)=\Pb(S\cap A)+\Pb(S\cap\no{A})$.   

  We define $\A=\{A~|~A$ is measurable$\}$.
\end{definition}

\section{A probability space of strategies}\label{sec:prob}
The aim of this section is to prove that $(\Omega,\A,\Pb)$ is a probability space.
That is, the sample space $\Omega$ (the set of all possible strategies), the set of events $\A$, which is the set of the Lebesgue measurable sets of strategies, and the probability measure $\Pb$, form a probability space. Our proof follows~\cite{BenedettoCzaja09}.
We proceed by proving that this triplet satisfies the Kolmogorov axioms, that is the probability of any event is between $0$ and $1$, the probability of $\Omega$ is $1$, and the probability of any countable sequence of pairwise disjoint (that is incompatible) events, is the sum of their probabilities.
In order to do so, we need first to prove several properties.

Lemma~\ref{lem:LebesgueProps} establishes several known properties of Lebesgue measurable sets.

\Needspace*{2\baselineskip}
\begin{lemma}\label{lem:LebesgueProps}~
  \begin{enumerate}
    \item\label{lem:LebesgueProps:it:disjointUnion} Let $A\in\A$ and $S\in\parts{\Omega}$. If $A\cap S=\emptyset$, then $\Pb(A\cup S)=\Pb(A)+\Pb(S)$.
    \item\label{lem:LebesgueProps:it:inclusion} Let $A_1,A_2\in\A$. If $A_1\subseteq A_2$, then $\Pb(A_1)\leq\Pb(A_2)$.
    \item\label{lem:LebesgueProps:it:probZeroMeasurable} $\emptyset$, the empty set, is Lebesgue measurable.
    \item\label{lem:LebesgueProps:it:complement} $A$ is Lebesgue measurable if and only if $\no{A}$ is Lebesgue measurable.
    \item\label{lem:LebesgueProps:it:unionOfLebesgueIsLebesgue} If $A_1,A_2$ are Lebesgue measurable, then $A_1\cup A_2$ is Lebesgue measurable.
  \end{enumerate}
\end{lemma}
\begin{proof}~
  \begin{enumerate}
    \item $\Pb(A\cup S)=\Pb((A\cup S)\cap A)+\Pb( (A\cup S)\cap\no{A})=\Pb(A)+\Pb(S)$.

    \item First notice that by definition, $\Pb(S)\geq 0$ for any $S\in\parts{\Omega}$. Hence,
      $\Pb(A_2)=\Pb(A_2\cap A_1)+\Pb(A_2\cap\no{A_1})=\Pb(A_1)+\Pb(A_2\cap\no{A_1})\geq\Pb(A_1)$.

    \item Notice that $\Pb(\emptyset)=0$. Then, $\forall S\in\parts{\Omega}$, $\Pb(S)=\Pb(S\cap\emptyset)+\Pb(S\cap\Omega)$, so $\emptyset$ is Lebesgue measurable.

    \item Let $A$ be Lebesgue measurable, then $\forall S\subseteq\Omega$, $\Pb(S)=\Pb(S\cap A)+\Pb(S\cap\no{A})=\Pb(S\cap\no{A})+\Pb(S\cap {A}^{\sim\sim})$, so $\no{A}$ is Lebesgue measurable.

    \item Let $A_1,A_2$ be Lebesgue measurable, then $\forall S\subseteq\Omega$, we have

      \noindent\parbox{0.4\textwidth}{
	\begin{equation}\label{lem:LebesgueProps:eq:A1}
	  \Pb(S)=\Pb(S\cap A_1)+\Pb(S\cap\no{A_1})
	\end{equation}
      }
      \hfill and\hfill
      \parbox{0.4\textwidth}{
	\begin{equation}\label{lem:LebesgueProps:eq:A2}
	  \Pb(S)=\Pb(S\cap A_2)+\Pb(S\cap\no{A_2})
	\end{equation}
      }\\[-1ex]
      From set theory\hfill
      \parbox{0.7\textwidth}{
	\begin{equation}\label{lem:LebesgueProps:eq:setTheory}
	  S\cap(A_1\cup A_2)=
	  S\cap(A_1\cup(\no{A_1}\cap A_2))=
	  (S\cap A_1)\cup(S\cap\no{A_1}\cap A_2)
	\end{equation}
      }

      Using $S\cap\no{A_1}$ for $S$ in (\ref{lem:LebesgueProps:eq:A2}) gives
      \begin{equation}
	\label{lem:LebesgueProps:eq:A4}
	\Pb(S\cap\no{A_1})
	=\Pb(S\cap\no{A_1}\cap A_2)+\Pb(S\cap\no{A_1}\cap\no{A_2})
	=\Pb(S\cap\no{A_1}\cap A_2)+\Pb(S\cap\no{(A_1\cup A_2)})    
      \end{equation}
      From (\ref{lem:LebesgueProps:eq:setTheory}), using items~\ref{lem:LebesgueProps:it:disjointUnion} and \ref{lem:LebesgueProps:it:inclusion}, we have
      $\Pb(S\cap(A_1\cup A_2))=\Pb(S\cap A_1)+\Pb(S\cap\no{A_1}\cap A_2)$.
      Adding $\Pb(S\cap\no{(A_1\cup A_2)})$ to both sides gives
      $\Pb(S\cap(A_1\cup A_2)+\Pb(S\cap\no{(A_1\cup A_2)})
      =\Pb(S\cap A_1)+\Pb(S\cap\no{A_1}\cap A_2)+\Pb(S\cap\no{(A_1\cup A_2)})$
      Using (\ref{lem:LebesgueProps:eq:A4}) and (\ref{lem:LebesgueProps:eq:A1}) we obtain
      $\Pb(S\cap(A_1\cup A_2))+\Pb(S\cap\no{(A_1\cup A_2)})
      =\Pb(S\cap A_1)+\Pb(S\cap\no{A_1})=\Pb(S)$.
      \qedhere
  \end{enumerate}
\end{proof}

The concept of algebra (Definition~\ref{def:algebra}) gives a closure property of subsets. As a corollary of the Lemma~\ref{lem:LebesgueProps} we can show that the set $\mathcal{A}$ of Lebesgue measurable sets form an algebra (Corollary~\ref{cor:AAlg}).

\begin{definition}[Algebra]\label{def:algebra}
  Let $X$ be a set. We say that a set $\mathbb{A}\in\parts{X}$ is an algebra over $X$ if for all $A, B\in\mathbb{A}$, $A\cup B$, $\no{A}$ and $X$ itself are also in $\mathbb{A}$.
\end{definition}

\begin{corollary}\label{cor:AAlg}
  $\A$ is an algebra over $\Omega$.
\end{corollary}
\begin{proof}
  $\A\in\parts{\Omega}$. Let $A,B\in\A$, then by Lemma~\ref{lem:LebesgueProps}(\ref{lem:LebesgueProps:it:unionOfLebesgueIsLebesgue}), $A\cup B\in\A$. By Lemma~\ref{lem:LebesgueProps}(\ref{lem:LebesgueProps:it:complement}), $\no{A}\in\A$. Finally, by Lemma~\ref{lem:LebesgueProps}(\ref{lem:LebesgueProps:it:probZeroMeasurable}) and (\ref{lem:LebesgueProps:it:complement}), $\Omega\in\A$.
\end{proof}

Moreover, we can show that $A$ is a $\sigma$-algebra, that is an algebra, completed to include countably infinite operations. Definition~\ref{def:salgebra} formalises it.

\begin{definition}[$\sigma$-algebra]\label{def:salgebra}
  Let $X$ be a set. We say that a set $\Sigma\in\parts{X}$ is a $\sigma$-algebra over $X$ if it is an algebra and it is closed under countable unions, that is, if
  $A_1, A_2, A_3,\dots$ are in $\Sigma$, then so is $\bigcup A_i$.
\end{definition}

Theorem \ref{thm:AsAlg} states that the set $\A$ of Lebesgue measurable sets is a $\sigma$-algebra. We need to prove two properties of Lebesgue measurable sets first (Lemmas~\ref{lem:sumOfUnion} and~\ref{lem:inftySumOfUnion}).

\begin{lemma}\label{lem:sumOfUnion}
  Let $S\subseteq\Omega$ and $A_1,\dots,A_n\in\A$ be a disjoint family. Then
  \[
    \Pb\left(S\cap\left(\bigcup_{i=1}^n A_i\right)\right)=\sum_{i=1}^n\Pb(S\cap A_i)\enspace.
  \]
\end{lemma}
\begin{proof}
  We proceed by induction on $n$. If $n=1$ it is trivial. Assume it is true for $n-1$.
  Notice that

  \parbox{0.4\textwidth}{
    \begin{equation}\label{lem:sumOfUnion:eq:int-one}
      S\cap\left(\bigcup_{i=1}^n A_i\right)\cap A_n=S\cap A_n
    \end{equation}
  }
  \hfill and\hfill
  \parbox{0.4\textwidth}{
    \begin{equation}\label{lem:sumOfUnion:eq:int-two}
      S\cap\left(\bigcup_{i=1}^n A_i\right)\cap \no{A_n}=S\cap\left(\bigcup_{i=1}^{n-1} A_i\right)
    \end{equation}
  }

  Equation (\ref{lem:sumOfUnion:eq:int-one}) is clear, and (\ref{lem:sumOfUnion:eq:int-two}) follows since $\left(\bigcup_{i=1}^n A_i\right)\cap\no{A_n}=\bigcup_{i=1}^n(A_i\cap\no{A_n})=(\bigcup_{i=1}^{n-1}(A_i\cap\no{A_n}))\cup(A_n\cap \no{A_n})=\bigcup_{i=1}^{n-1}(A_i\cap\no{A_n})$.

  Thus, since $A_n$ is measurable, we have that
  $$\Pb\left(S\cap\left(\bigcup_{i=1}^n A_i\right)\right) = \Pb\left(S\cap\left(\bigcup_{i=1}^n A_i\right)\cap A_n\right)+\Pb\left(S\cap\left(\bigcup_{i=1}^n A_i\right)\cap\no{A_n}\right)$$
  and from (\ref{lem:sumOfUnion:eq:int-one}) and (\ref{lem:sumOfUnion:eq:int-two}) this is equal to $\Pb(S\cap A_n)+\Pb\left(S\cap\left(\bigcup_{i=1}^{n-1} A_i\right)\right)$, which by the induction hypothesis is equal to $\sum_{i=1}^n\Pb(S\cap A_i)$.
\end{proof}

\begin{lemma}\label{lem:inftySumOfUnion}
  Let $S_1,S_2,\dots\subseteq\Omega$. Then
  \[
    \Pb\left(\bigcup_{i=1}^\infty S_i\right)\leq\sum_{i=1}^\infty\Pb(S_i)\enspace.
  \]
\end{lemma}

\begin{proof}
  If $\Pb(S_i)=\infty$ for some $i$, then we are finished. Therefore, assume $\Pb(S_i)<\infty$ for each $i\in\mathbb{N}$.

  Without lost of generality, assume $S_i\neq\emptyset$, for all $i$. Indeed, since $\Pb(\emptyset)=0$, an empty set would not add anything to any side of the equation.
  For a given $\varepsilon >0$ and $i$, there is a sequence $\{B_{ij}~|~i=1,\dots, j=1,\dots\}$ of boxes such that $S_i\subseteq\bigcup_{j=1}^\infty B_{ij}$ and $\sum_{j=1}^\infty \p(B_{ij})<\Pb(S_i)+2^{-i}\varepsilon$, by the definition of $\Pb$. Now, $\#\{B_{ij}~|~i,j\}\leq\aleph_0$ and $\bigcup_{i=1}^\infty S_i\subseteq\bigcup_{i=1}^\infty\bigcup_{j=1}^\infty B_{ij}$. Therefore, using the definition of \Pb,
  $$ \Pb\left(\bigcup_{i=1}^\infty S_i\right)\leq\sum_{i=1}^\infty\sum_{j=1}^\infty \p(B_{ij})
  \leq \sum_{i=1}^\infty\Pb(S_i)+\varepsilon\sum_{i=1}^\infty\frac{1}{2^i}=\sum_{i=1}^\infty\Pb(S_i)+\varepsilon$$
  Since this is true for each $\varepsilon$, the lemma holds.
\end{proof}

Using these properties, we can prove that $\A$ is a $\sigma$-algebra (Theorem~\ref{thm:AsAlg}).

\begin{theorem}\label{thm:AsAlg}
  $\A$ is a $\sigma$-algebra over $\Omega$.
\end{theorem}
{\begin{proof}
  By Corollary~\ref{cor:AAlg}, $\A$ is an algebra. 
  We only have to prove that $\A$ is closed under any countable unions. That is, if $B_1,B_2,\dots\in\A$, then $\bigcup_{i=1}^{\infty} B_i\in\A$.
  Since $\A$ is an algebra (Corollary~\ref{cor:AAlg}), there is a disjoint family $A_1,A_2,\dots\in\A$ such that
  $A=\bigcup_{i=1}^{\infty}B_i=\bigcup_{i=1}^{\infty}A_i$.
  For example, we can take $A_1=B_1$, $A_2=B_2\setminus B_1, A_3=B_3\setminus(B_1\cup B_2),\dots$.
  Let $C_n=\bigcup_{i=1}^n A_i$, so $C_n\in\A$ again using that $\A$ is an algebra. Also notice that $\no{A}\subseteq\no{C_n}$ because $C_n\subseteq A$.

  Since $C_n$ is measurable, take any $S\subseteq\Omega$ and, using Lemma~\ref{lem:LebesgueProps}(\ref{lem:LebesgueProps:it:inclusion}), we can calculate
  $\Pb(S) = \Pb(S\cap C_n)+\Pb(S\cap \no{C_n})\geq\Pb(S\cap C_n)+\Pb(S\cap \no{A})$.
  Since $\Pb(S\cap C_n)=\Pb\left(S\cap\left(\bigcup_{i=1}^n A_i\right)\right)$, using Lemma~\ref{lem:sumOfUnion}, we obtain
  $\Pb(S)\geq\sum_{i=1}^n\Pb(S\cap A_i)+\Pb(S\cap\no{A})$
  and, since the left-hand side is independent of $n$,
  $\Pb(S)\geq\sum_{i=1}^\infty\Pb(S\cap A_i)+\Pb(S\cap\no{A})$.
  Thus, by Lemma~\ref{lem:inftySumOfUnion},
  $\Pb(S)\geq\Pb\left(S\cap\left(\bigcup_{i=1}^\infty A_i\right)\right)+\Pb(S\cap\no{A})=\Pb(S\cap A)+\Pb(S\cap\no{A})$.

  For the converse inequality, notice that $S=(S\cap A)\cup(S\cap\no{A})$, so using Lemma~\ref{lem:inftySumOfUnion} we have
  $\Pb(S)=\Pb((S\cap A)\cup(S\cap\no{A}))\leq\Pb(S\cap A)+\Pb(S\cap\no{A})$.
  Hence, $A\in\A$.
\end{proof}}

As intuited in Example~\ref{ex:probOfBox}, the probability of a box $B$ is $\p(B)$. Lemma~\ref{lem:Peqp} formalises it. Before proving this lemma, we need two auxiliary ones (Lemmas~\ref{lem:smallpintersection} and \ref{lem:smallpintersections}). For short, we use the notation $B\cap\ve a=\ve b$ for $B\cap\{f~|~f(\ve a)=\ve b\}$.

\begin{lemma}\label{lem:smallpintersection}
  Let $N\subseteq\mathbb{N}$ and for all $i\in N$, let $B, B_i\subseteq\Omega$ be boxes s.t. $B\subseteq\bigcup_{i\in N}B_i$ and $\p(B)>\sum_{i\in N}\p(B_i)$. Then for every object $\ve a$, there exists an object $\ve b$ such that,
  $\p(B\cap\ve a=\ve b)>\sum_{i\in N}\p(B_i\cap\ve a=\ve b)$.
\end{lemma}
\begin{proof}
  Let $\ve a\to\ve b_i$, with $i=1,\dots,n$. Hence notice that $\p(B)=\sum_{j=1}^n \p(B\cap\ve a=\ve b_j)$, and this happens for any $B$. Then, from $\p(B)>\sum_{i\in N}\p(B_i)$, we have
  $\sum_{j=1}^n \p(B\cap\ve a=\ve b_j)
  > \sum_{i\in N}\sum_{j=1}^n \p(B_i\cap\ve a=\ve b_j)\\
  =\sum_{j=1}^n \sum_{i\in N}\p(B_i\cap\ve a=\ve b_j)$.
  Therefore, there must be at least one $h$ such that
  $\p(B\cap\ve a=\ve b_h)>\sum_{i\in N}\p(B_i\cap\ve a=\ve b_h)$.
\end{proof}

\begin{lemma}\label{lem:smallpintersections}
  Let $N\subseteq\mathbb{N}$ and for all $i\in N$, let $B, B_i\subseteq\Omega$ be boxes s.t. $B\subseteq\bigcup_{i\in N}B_i$ and $\p(B)>\sum_{i\in N}\p(B_i)$. Then for all family $\{\ve a_j\}$ of objects, there exists a family $\{\ve b_j\}$ such that, for every $k$,\
  $\p(B\cap\ve a_1=\ve b_1\cap\cdots\cap\ve a_k=\ve b_k)>\sum_{i\in N}\p(B_i\cap\ve a_1=\ve b_1\cap\cdots\cap\ve a_k=\ve b_k)$.
\end{lemma}
\begin{proof}
  We proceed by induction on $k$. For $k=1$, use Lemma~\ref{lem:smallpintersection}.
  By the induction hypothesis, we have
  $\p(B\cap\ve a_1=\ve b_1\cap\cdots\cap\ve a_{k-1}=\ve b_{k-1})>\sum_{i\in N}\p(B_i\cap\ve a_1=\ve b_1\cap\cdots\cap\ve a_{k-1}=\ve b_{k-1})$. We conclude by Lemma~\ref{lem:smallpintersection}.
\end{proof}

\begin{lemma}\label{lem:Peqp}
  Let $B\subseteq\Omega$ be a box, then $\Pb(B)=\p(B)$.
\end{lemma}
\begin{proof}
  Let $B=\{f~|~f(\ve a_1)=\ve a_1',\dots,f(\ve a_n)=\ve a_n'\}$. Since $B\subseteq B$, by definition of $\Pb$, we have $\Pb(B)\leq \p(B)$.
  We must prove $\p(B)\leq\Pb(B)=\inf\{\sum_{i\in N}\p(B_i)~|~B\subseteq\bigcup_{i\in N}B_i\}$. In other words, we must prove that $B\subseteq\bigcup_{i\in N}B_i$ implies $\p(B)\leq \sum_{i\in N}\p(B_i)$.
  We proceed by induction on $n$.
  \begin{itemize}
    \item If $n=0$, $\p(B)=1$. Notice that, without restrictions in $B$, $B=\Omega$. We prove this case by contradiction. Let $\p(F)>\sum_{i\in N}\p(B_i)$. Then by Lemma~\ref{lem:smallpintersections}, there exists $g$ such that for all $k$,
      \begin{equation}\label{eq:one}
	\p(\ve a_1=g(\ve a_1)\cap\cdots\cap\ve a_k=g(\ve a_k))>\sum_{i\in N}\p(B_i\cap \ve a_1=g(\ve a_1)\cap\cdots\cap\ve a_k=g(\ve a_k))
      \end{equation}
      Since $g\in\Omega\subseteq\bigcup_{i\in N}B_i$, there exists $j$ such that $g\in B_j$. Let $B_j$ be defined with constraints on objects $\ve a_{j_1},\dots,\ve a_{j_q}$. Let $k=q$ and from equation (\ref{eq:one}),
      \begin{equation}\label{eq:two}
	\p(\ve a_1=g(\ve a_1)\cap\cdots\cap\ve a_q=g(\ve a_q))>\sum_{i\in N}\p(B_i\cap \ve a_1=g(\ve a_1)\cap\cdots\cap\ve a_q=g(\ve a_q))
      \end{equation}
      We know that 
      $\p(\ve a_1=g(\ve a_1)\cap\cdots\cap\ve a_q=g(\ve a_q))={\prod_{h=1}^q\frac{\to(\ve a_h,g(\ve a_h))}{\rho(\ve a_h)}}$, and since $g\in B_j$, we know that this is also equal to
      $\p(B_j\cap \ve a_1=g(\ve a_1)\cap\cdots\cap\ve a_q=g(\ve a_q))$. Hence equation (\ref{eq:two}) leads to a contradiction.

    \item Consider the case $n-1$. Let $B'=\{f~|~\exists g\in B$ s.t. $\forall\ve a\neq\ve a_n,\ f(\ve a)=g(\ve a)\}$. Then if $B'\subseteq\bigcup_{i\in N'}B'_i$ we have $\p(B')\leq\sum_{i\in N'}\p(B'_i)$.
      Notice that either $B_i'=B_i$ or $B_i$ has a constraint on $\ve a_n$ and so $\frac{\to(\ve a_n,g(\ve a_n))}{\rho(\ve a_n)}\p(B_i')=\p(B_i)$. In any case, $\frac{\to(\ve a_n,g(\ve a_n))}{\rho(\ve a_n)}\p(B'_i)\leq \p(B_i)$.
      Then
      $\p(B)=\frac{\to(\ve a_n,g(\ve a_n))}{\rho(\ve a_n)}\p(B')\leq\sum_{i\in N'}\frac{\to(\ve a_n,g(\ve a_n))}{\rho(\ve a_n)}\p(B'_i)\leq\sum_{i\in N'}\p(B_i)$.\qedhere
  \end{itemize}
\end{proof}

\begin{theorem}[Space of strategies]\label{thm:probabilitySpace}
  $(\Omega,\A,\Pb)$ is a probability space.
\end{theorem}
\begin{proof} We prove it satisfies the Kolmogorov axioms.
  \begin{description}
    \item[1$^{\mbox{st}}$ axiom:] $\forall A\in\A$, $0\leq \Pb(A)\leq 1$.

      Since $\Pb$ is defined as an $\inf$ of sums of $p$, and $p$ is always positive, so $\Pb$ cannot be negative. By the second Kolmogorov axiom $\Pb(\Omega)=1$. Notice that $A$ is measurable and $A\subseteq\Omega$, so $1=\Pb(\Omega)=\Pb(\Omega\cap A)+\Pb(\Omega\setminus A)=\Pb(A)+\Pb(\Omega\setminus A)$, hence $1-\Pb(\Omega\setminus A)=\Pb(A)$. Since $\Pb$ is not negative, $\Pb(A)\leq 1$.

    \item[2$^{\mbox{nd}}$ axiom:] $\Pb(\Omega)=1$.

      Notice that $\Omega$ is the box including all the functions. Hence, there is no condition on the functions and so $n=0$. Then $\p(\Omega)=1$. By Lemma~\ref{lem:Peqp}, $\Pb(\Omega)=\p(\Omega)=1$.

    \item[3$^{\mbox{rd}}$ axiom:] Any countable sequence of pairwise disjoint (i.e.~incompatible) events $A_1,A_2\dots\in\A$, satisfies $\Pb(A_1\cup A_2\dots)=\sum_{i=1}^\infty \Pb(A_i)$.

      Let $\emptyset\neq I\subsetneq\mathbb{N}$. Since the sets $A_i$ are in $\A$, consider $n\in\mathbb{N}\setminus I$ and we have
      $$\Pb\left(\bigcup_{i\in\mathbb{N}\setminus I} A_i\right)=\Pb\left(\left(\bigcup_{i\in\mathbb{N}\setminus I} A_i\right)\cap A_n\right)+\Pb\left(\left(\bigcup_{i\in\mathbb{N}\setminus I} A_i\right)\cap\no{A_n}\right)$$
      Notice that $\left(\bigcup_{i\in\mathbb{N}\setminus I} A_i\right)\cap A_n=A_n$ and since the $A_i$'s are pairwise disjoint $\left(\bigcup_{i\in\mathbb{N}\setminus I} A_i\right)\cap\no{A_n}=\bigcup_{i\in\mathbb{N}\setminus(I\cup\{n\})} A_i$. Therefore, considering that this is valid for any $I$ and $n\notin I$, we have
      \begin{displaymath}
	\Pb\left(\bigcup_{i=1}^\infty A_i\right)=\Pb(A_1)+\Pb\left(\bigcup_{i=2}^\infty A_i\right)
	= \Pb(A_1)+\Pb(A_2)+\Pb\left(\bigcup_{i=3}^\infty A_i\right)
	=\dots = \sum_{i=1}^\infty \Pb(A_i). \qedhere 
      \end{displaymath}
  \end{description}
\end{proof}

\begin{example}\label{ex:long}
  Consider the non-strongly-normalising non-confluent rewrite system described in the introduction
  $\ve a_i\to\ve a_{i+1},\quad\ve a_i\to\ve a'_{i+1}$, where each reduction is equiprobable and each symbol is different from each other. It can be depicted as follows.
  $$
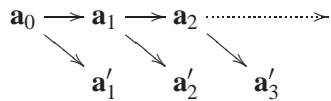
\xymatrix@C=3ex@R=2ex{
    \ve a_0 \ar[r]\ar[dr] & \ve a_1 \ar[dr]\ar[r] & \ve a_2\ar[dr]\ar@{..>}[rr]& & \\
    & \ve a'_1 & \ve a'_2 & \ve a'_3 &
  }$$

  The probability that this rewrite system stops after exactly $n$ steps, starting from term $\ve a_0$ is $\Pb(B)$, with $B=\{f~|~f(\ve a_0)=\ve a_1,\dots f(\ve a_{n-2})=\ve a_{n-1}$ and $f(\ve a_{n-1})=\ve a'_n\})$, and since $B$ is a box, by Lemma~\ref{lem:Peqp} it is the same to $\Pb(B)=\p(B)=\dfrac{1}{\rho(\ve a_0)\dots\rho(\ve a_{n-1})}=\dfrac{1}{2^n}$.

  The probability of stopping at the step $n$ or before, starting at any point before $a_{n-1}$, is just the probability of the box $\{f~|~f(\ve a_{n-1})=\ve a'_n\}$, which is $\dfrac{1}{2}$.

  The probability of stopping at the step $n$ or $m$, starting at any point before $\ve a_{n-1}$ and $\ve a_{m-1}$ is the probability of the union of two boxes, however they are not independent events (its intersection is not empty). Hence let $B_1=\{f~|~\ve a_{n-1}=\ve a'_{n}\}$ and $B_2=\{f~|~\ve a_{m-1}=\ve a'_{m}\}$. The probability $\Pb(B_1\cup B_2)=\Pb((B_1\setminus B_2)\cup B_2)=\Pb(B_1\setminus B_2)\cup \Pb(B_2)=\Pb(\{f~|~\ve a_{n-1}=\ve a'_{n},\ve a_{m-1}=\ve a'_{m})+\Pb(B_2)=\dfrac{1}{4}+\dfrac{1}{2}=\dfrac{3}{4}$.

  Finally, the probability of not stopping at all, is the probability of the set $S=\{f~|~f(\ve a_i)=\ve a_{i+1}$ for $i\in\mathbb{N}\}$, which is not a box, since there is an infinite number of conditions. It is easy to check that we need an infinite number of boxes to cover such a set, however we can chose boxes as small as we want (that is, with a big number of conditions), which makes the infimum of their sums to be $0$, and so the probability of not stopping is, as expected, $0$.

  In other words, $\Pb(S)\leq\{f~|~f(\ve a_i)=\ve a_{i+1}, i\in[0,n]\}=\frac{1}{2^n}$, for any $n$. Hence when $n$ tends to $\infty$, $\Pb(S)$ tends to $0$.
\end{example}

\section{Transforming a non-deterministic into a probabilistic calculus}\label{sec:Lpplus}
\subsection{The calculus \texorpdfstring{\Lplus}{L+}}

In \cite{DiazcaroDowekLSFA12,DiazcaroDowek13} we have introduced a non-deterministic calculus called \Lplus, which is a simplification of an earlier probabilistic calculus by keeping non-determinism but removing explicit probabilities. Now we can transform this calculus into a probabilistic one.

The full calculus is depicted in Table~\ref{tab:Lplus}.
Typing judgements are of the form $\ve r:A$. A term $\ve r$ is typable if there exists a type $A$ such that $\ve r:A$. 
Following~\cite{GeuversKrebbersMcKinnaWiedijkLFMTP10,ParkSeoParkLeeJAR13}, 
we use a presentation of typed lambda-calculus without contexts and where
each variable occurrence is labelled by its type, such as.
$\lambda x^A.x^A$ or $\lambda x^A.y^B$. 
We sometimes omit the labels when they are clear from the context and 
write, for example, $\lambda x^A.x$ for $\lambda x^A.x^A$.
We use different letters for different variables and the type system forbids terms 
such as $\lambda x^A.x^B$ when $A$ and $B$ are different, by imposing preconditions to when the typing rules apply.
Let $S=\{x_1^{A_1},\dots,x_n^{A_n}\}$ be a set of declarations, we write 
$\coherentContext S$ when this set is functional, that is when $x_i = x_j$ implies 
$A_i = A_j$. 
For example $\coherentContext{\{x^A,y^{A\Rightarrow B}\}}$, but not 
$\coherentContext{\{x^A,x^{A\Rightarrow B}\}}$.
Typing rules have the following structure:
\[
  \condi{\cond{\mbox{\scriptsize Preconditions}}}{\infer[^{\mbox{\scriptsize (Rule name)}}]{\mbox{Derived judgement}}{\mbox{Hypotheses}}}
\]

The $\alpha$-conversion and the sets $FV(\ve r)$ of free variables of $\ve r$ and $FV(A)$ of free variables of $A$ are defined as usual in the $\lambda$-calculus (cf.~\cite[\S2.1]{Barendregt84}). For example $FV(x^Ay^B)=\{x^A,y^B\}$.
We say that a term $\ve r$ is closed whenever $FV(\ve r)=\emptyset$.
If $FV(\ve r)=\{x_1^{A_1},\dots,x_n^{A_n}\}$, we write $\Gamma(\ve r)=\{A_1,\dots,A_n\}$. $FV(\{A_1,\dots,A_n\})$ is defined by $\bigcup_{i=1}^nFV(A_i)$.
Given two terms $\ve r$ and $\ve s$ we denote by $\ve r[\ve s/x]$ the term obtained by simultaneously substituting the term $\ve s$ for all the free occurrences of $x$ in $\ve r$, subject to the usual proviso about renaming bound variables in $\ve r$ to avoid capture of the free variables of $\ve s$. Analogously $A[B/X]$ denotes the substitution of the type $B$ for all the free occurrences of $X$ in $A$, and $\ve r[B/X]$ the substitution in $\ve r$. For example, $(x^A)[B/Y]=x^{(A[B/Y])}$, $(\lambda x^A.\ve r)[B/X]=\lambda x^{A[B/X]}.\ve r[B/X]$ and $(\pi_A(\ve r))[B/X]=\pi_{A[B/X]}(\ve r[B/X])$. Simultaneous substitutions are defined in the same way. Finally, terms and types are considered up to $\alpha$-conversion.

Each term of the language has a main type associated, which can be obtained from the type annotations, and other types induced by the type equivalences.

The operational semantics of \Lplus\ is also given in Table~\ref{tab:Lplus}, where there are two distinct relations between terms: a symmetric relation $\eq$ and a reduction relation $\re$.
We write $\eq^*$ and $\re^*$ for the transitive and reflexive closures of $\eq$ and $\re$ respectively.
In particular, notice that $\eq^*$ is an equivalence relation. We just write $\to$ when we do not want to make the distinction between these relations.
We write $n.\ve r$ in \Lplus~as a shorthand for $\underbrace{\ve r+\cdots+\ve r}_{n\mbox{ times}}$.

This calculus has a non-deterministic projector. Indeed, the rule ``If $\ve r:A$, then $\pi_A(\ve r+\ve s)\re\ve r$'' is not-deterministic because the symbol $+$ is commutative, so if $\ve s:A$, this rule can produce either $\ve r$ or $\ve s$ non-deterministically. In any case, both reducts are valid proofs of $A$, and so the proof system is consistent. Refer to \cite{DiazcaroDowekLSFA12} for details.

{\begin{table}[!ht] 
  {\centering
    {\bf Grammar of types and terms}
    \begin{align*}
      A,B,C,\dots\ &::=\ X~|~A\Rightarrow B~|~A\wedge B~|~\forall X.A\enspace.\\
      \ve r,\ve s,\ve t\ &::=\ x^A~|~\lambda x^A.\ve r~|~\ve r\ve s~|~\ve r+\ve s~|~\pi_{A}(\ve r)~|~\Lambda X.\ve r~|~\ve r\{A\}\enspace.
    \end{align*}

    {\bf Equivalence between types}
    $$A\wedge B~\equiv~ B\wedge A\enspace,\quad\quad(A\wedge B)\wedge C~\equiv~ A\wedge(B\wedge C)\enspace,\quad\quad A\Rightarrow (B\wedge C)~\equiv~ (A\Rightarrow B)\wedge(A\Rightarrow C)\enspace.$$

    {\bf Rewriting system}

    \emph{Symmetric relation:}
    \begin{tabular}{c@{\hspace{1cm}}c@{\hspace{1cm}}c}
      $\ve r+\ve s\eq\ve s+\ve r$\enspace, & 
      $(\ve r+\ve s)\ve t\eq\ve r\ve t+\ve s\ve t$\enspace, &
      If $\ve r:A\Rightarrow (B\wedge C)$, then \\
      $(\ve r+\ve s)+\ve t\eq\ve r+(\ve s+\ve t)$\enspace, &
      $\lambda x^A.(\ve r+\ve s)\eq\lambda x^A.\ve r+\lambda x^A.\ve s$\enspace, &
      $\pi_{A\Rightarrow B}(\ve r)\ve s\eq\pi_B(\ve r\ve s)$\enspace.
    \end{tabular}
    \vspace{0.3cm}

    \emph{Reductions:}\\
    \hfill$(\lambda x^A.\ve r)~\ve s\re\ve r[\ve s/x]$\enspace,\hfill
    $(\Lambda X.\ve r)\{A\}\re\ve r[A/X]$\enspace,\hfill
    If $\ve r:A$, then $\pi_A(\ve r+\ve s)\re\ve r$\enspace.\hfill\enspace\\
    \medskip

    {\bf Typing system}
    \medskip

    \hfill
    \condi{\cond{A\equiv B}}{\infer[^{(\equiv)}]{\ve r:B}{\ve r:A}}
    \hfill
    \condi{}{\infer[^{(ax)}]{x^A:A}{}}
    \hfill
    \condi{\cond{\coherentContext{(FV(\ve r)\cup\{x^A\})}}}
    {\infer[^{(\Rightarrow_i)}]
    {\lambda x^A.\ve r:A\Rightarrow B}
    {\ve r:B}}
    \hfill 
    \condi{\cond{\coherentContext{FV(\ve r\ve s)}}}
    {
      \infer[^{(\Rightarrow_e)}]
      {\ve r\ve s:B}
      {\ve r:A\Rightarrow B & \ve s:A}
    }
    \hfill
    \medskip

    \hfill
    \condi{\cond{\coherentContext{FV(\ve r+\ve s)}}}
    {
      \infer[^{(\wedge_i)}]
      {\ve r+\ve s:A\wedge B}
      {\ve r:A & \ve s:B}
    }
    \hfill
    \condi{}{\infer[^{(\wedge_e)}]{\pi_A(\ve r):A}{\ve r:A\wedge B}}
    \hfill
    \condi{\cond{X\notin FV(\Gamma(\ve r))}}
    {\infer[^{(\forall_i)}]
    {\Lambda X.\ve r:\forall X.A}
    {\ve r: A}}
    \hfill
    \condi{}{\infer[^{(\forall_e)}]{\ve r\{B\}:A[B/X]}{\ve r:\forall X.A}}
    \hfill~

    \vspace{0.1cm}
  }
  \caption{The non-deterministic calculus \Lplus}
  \label{tab:Lplus}
\end{table}}

\subsection{From non-determinism to probabilities (or from \texorpdfstring{\Lplus}{L+}~to \texorpdfstring{\Lpplus}{L+p})}

Consider the following example (cf.\ \cite[Example 5]{DiazcaroDowek13}).
Two possible reduction paths can be fired from $(\Lambda X.(\pi_A(x^A + y^X )))\{A\}$:
Reducing first the projection,
$(\Lambda X.x^A)\{A\}\re x^A$, 
or
reducing first the beta
$\pi_A(x^A+y^A)\re x^A$.
The former path is deterministic and will always reduce to $x^A$, on the contrary, the latter can non-deterministically chose between $x^A$ and $y^A$. However, in both cases a proof of $A$ is obtained. 

Hence, the non-determinism is present not only due to the projector,
but also by a combination of not defining a reduction strategy and the
polymorphism, which can turn a deterministic projection into a
non-deterministic one.
We want to associate a probability to the second case, that is, to the non-deterministic projector (the $\pi$ reduction). 
With this aim, we consider the following ARS, called $\Lplus^\downarrow$.
The closed normal terms of \Lplus\ are objects of $\Lplus^\downarrow$. If $\ve r_1,\dots,\ve r_n$ are objects, 
then it is also an object.
The function $\to$ is given by the relations $\eq$ and $\re$.
In particular, if $\ve r:A$, then $\pi_A(\ve r+\ve r)\to\ve r$, with multiplicity $2$, i.e. $\to(\pi_A(\ve r+\ve r),\ve r)=2$.

\begin{theorem}\label{thm:toprob}
  Let $(\Omega,\A,\Pb)$ be a probability space over $\Lplus^\downarrow$.
  Let
  $B_{\ve r_i}=\{f~|~f(\pi_A(\suj{n}m_j.\ve r_j))=\ve r_i\}$ be a box.
  Then $\Pb(B_{\ve r_i})=\tfrac{m_i}{\suj{n}m_j}$.
\end{theorem}
\vspace{-9mm}
{\begin{proof}
  Notice that
    $
    \rho(\pi_A(\sui{n}m_i.\ve r_i)) 
    = \sum_{\ve r}\to(~\pi_A(\sui{n}m_i.\ve r_i),\ve r)
    = \sharp[\overbrace{\ve r_1,\cdots,\ve r_1}^{m_1\mbox{ times}},\dots,\overbrace{\ve r_n,\cdots,\ve r_n}^{m_n\mbox{ times}}]
    = \suj{n}m_j
    $
  And $\to(\pi_A(\sui{n}m_i.\ve r_i),\ve r_i)=m_i$.  Hence, $\Pb(B_{\ve r_i})=\p(B_{\ve r_i})=\frac{m_i}{\suj{n}m_j}$.
\end{proof}}

\begin{definition}[The probabilistic calculus \Lpplus]
  Let \Lpplus\ be the language of Table~\ref{tab:Lplus}, with the following modification:\\
  Replace rule ``If $\ve r:A$, then $\pi_A(\ve r+\ve s)\re\ve r$'' by\\
  ``For $i=1,\dots,n$, let $\ve r_i:A$ and $\ve s\notder A$, be closed normal terms. Then
  \begin{displaymath}
    \pi_A(\sui{n}m_i.\ve r_i+\ve s)\re\ve r_i\qquad\mbox{with probability }\frac{m_i}{\suj{n}m_j}\mbox{''}\enspace.   
  \end{displaymath}
\end{definition}

\begin{remark}
  Notice that by Theorem~\ref{thm:toprob} the probabilistic reduction is well defined.
\end{remark}

\subsection{The calculus \texorpdfstring{\Alg}{Alg}}

The calculus \Alg\ is inspired from \cite{ArrighiDowekRTA08,VauxMSCS09}. We restrict the algebraic calculus to only have probabilistic superpositions, and we type it with a simple extension of System $F$ (cf.~\cite[Def.~5.1]{ArrighiDiazcaroLMCS12}). The grammar of terms ensures that the linear combinations of terms are probability distributions, however the type system allows typing pseudo-terms, that is, terms that are not probability distributions. A term in this language, is a term produced by the grammar of terms, and typed. The full calculus is depicted in Table~\ref{tab:Alg}.

\begin{table}[ht]
  {\centering
    {\bf Grammar of types}
    $$A,B,C,\dots\ ::=\quad X~|~A\Rightarrow B~|~\forall X.A\enspace.$$
    {\bf Grammar of pseudo-terms}
    $$\ve r,\ve s,\ve t\ ::=
    \quad x^A~|~
    \lambda x^A.\ve r~|~
    \ve{rs}~|~\Lambda X.\ve r~|~
    \ve r\{A\}~|~
    p.\ve r~|~
    \ve r+\ve s$$
    {\bf Grammar of terms}
    $$\ve r,\ve s,\ve t\ ::=
    \quad x^A~|~
    \lambda x^A.\ve r~|~
    \ve{rs}~|~\Lambda X.\ve r~|~
    \ve r\{A\}~|~
    \sui{n} p_i.\ve r_i
    \qquad\mbox{with }
    \left\{
      \begin{array}{l}
	n>0,\\
      p_i\in\mathbb{Q}(0,1]\mbox{ and}\\
      \sui{n}p_i=1
    \end{array}
    \right.$$

    {\bf Rewriting system}

    \emph{Symmetric relation:}
    \begin{tabular}{c@{\qquad\qquad}c@{\qquad\qquad}c}
      $\ve r+\ve s\eq\ve s+\ve r$\enspace, & 
      $(\ve r+\ve s)\ve t\eq\ve r\ve t+\ve s\ve t$\enspace, &
      $1.\ve r\eq\ve r$\enspace.\\
      $(\ve r+\ve s)+\ve t\eq\ve r+(\ve s+\ve t)$\enspace, &
      $\lambda x^A.(\ve r+\ve s)\eq\lambda x^A.\ve r+\lambda x^A.\ve s$\enspace,&
    \end{tabular}
    \medskip

    \emph{Reductions:}
    \begin{tabular}{c@{\qquad\qquad}c@{\qquad\qquad}c}
      Beta & Elementary & Factorisation\\
      $(\lambda x^A.\ve r)~\ve s\re\ve r[\ve s/x]$\enspace, &
      $p.q.\ve r\re pq.\ve r$\enspace, &
      $p.\ve r+q.\ve r\re(p+q).\ve r$\enspace.\\
      $(\Lambda X.\ve r)\{A\}\re\ve r[A/X]$\enspace, &
      $p.(\ve r+\ve s)\re p.\ve r+p.\ve s$\enspace, &
    \end{tabular}
    \medskip

    {\bf Typing system}\\
    \medskip

    \hfill
    \condi{}{\infer[^{(ax)}]{x^A:A}{}}
    \hfill
    \condi{\cond{\coherentContext{(FV(\ve r)\cup\{x^A\})}}}
    {\infer[^{(\Rightarrow_i)}]
    {\lambda x^A.\ve r:A\Rightarrow B}
    {\ve r:B}}
    \hfill 
    \condi{\cond{\coherentContext{FV(\ve r\ve s)}}}
    {
      \infer[^{(\Rightarrow_e)}]
      {\ve r\ve s:B}
      {\ve r:A\Rightarrow B & \ve s:A}
    }
    \hfill
    \medskip

    \hfill
    \condi{\cond{\coherentContext{FV(\ve r+\ve s)}}}
    {
      \infer[^{(+_i)}]
      {\ve r+\ve s:A}
      {\ve r:A & \ve s:A}
    }
    \hfill
    \condi{}{\infer[^{(p_i)}]{p.\ve r:A}{\ve r:A}}
    \hfill
    \condi{\cond{X\notin FV(\Gamma(\ve r))}}
    {\infer[^{(\forall_i)}]
    {\Lambda X.\ve r:\forall X.A}
    {\ve r: A}}
    \hfill
    \condi{}{\infer[^{(\forall_e)}]{\ve r\{B\}:A[B/X]}{\ve r:\forall X.A}}
    \hfill~

    \vspace{0.1cm}
  }
  \caption{The algebraic calculus \Alg.}
  \label{tab:Alg}
\end{table}

\subsection{From \texorpdfstring{\Alg}{Alg}~to~\texorpdfstring{\Lpplus}{L+p}}
We give a translation from the probabilistic calculus $\Alg$, including scalars, to the probabilistic calculus $\Lpplus$.
\[
  \begin{array}{r@{~=~}l@{\qquad}r@{~=~}l@{\qquad}r@{~=~}l}
    \AtP{x^A} & x^A
    &
    \AtP{\ve r\ve s} & \AtP{\ve r}\AtP{\ve s}
    &
    \AtP{\ve r\{A\}} & \AtP{\ve r}\{A\}
    \\
    \AtP{\lambda x^A.\ve r} & \lambda x^A.\AtP{\ve r}
    &
    \AtP{\Lambda X.\ve r} & \Lambda X.\AtP{\ve r}
    &
    \AtP{\sui{n}\dfrac{n_i}{d_i}.\ve r_i} & \pi_A(\sui{n}m_i.\AtP{\ve r_i})\\[-2mm]
    \multicolumn{4}{c}{}
    &
    \multicolumn{2}{r}{
      \text{where }\ve r_i:A, d_i\in\mathbb{N}^*, m_i=n_i(\prod\limits_{\arriba{k=1}{k\neq i}}^{n}d_k),\text{ for }i=1,\dots,n.
    }
  \end{array}
\]

\begin{example}
  Let $\ve r:A$, $\ve t:A$ and $\ve s:A$.
  $\AtP{
    \dfrac{3}{4}.\ve r
    +\dfrac{1}{8}.\ve t
    +\dfrac{1}{8}.\ve s
  }
  =
  \pi_A
  \left(
  192.
  \AtP{\ve r}
  +
  32.
  \AtP{\ve t}
  +
  32.
  \AtP{\ve s}
  \right)$.
  By Theorem~\ref{thm:toprob}, this last term reduces to $\AtP{\ve r}$ with probability $\frac{192}{192+32+32}=\frac{3}{4}$, to 
  $\AtP{\ve t}$ with probability $\frac{32}{192+32+32}=\frac{1}{8}$, and to
  $\AtP{\ve s}$ with probability $\frac{32}{192+32+32}=\frac{1}{8}$.
\end{example}

\begin{lemma}\label{lem:AtPSubstitution}~
  \begin{multicols}{2}
  \begin{enumerate}
    \item\label{it:AtPSubsType}
      $\AtP{\ve r}[A/X]
      =
      \AtP{\ve r[A/X]}$.
    \item\label{it:AtPSubsTerm}
      $\AtP{\ve r}[\AtP{\ve s}/x]
      =
      \AtP{\ve r[\ve s/x]}$.
  \end{enumerate}
  \end{multicols}
\end{lemma}
\begin{proof}~
  \begin{enumerate}
    \item We proceed by induction on $\ve r$.
      \begin{itemize}
	\item Let $\ve r=x^B$.
	  $\AtP{x^B}[A/X]=
	  x^B[A/X]=
	  x^{B[A/X]}=
	  \AtP{x^{B[A/X]}}=
	  \AtP{x^B[A/X]}$.

	\item Let $\ve r=\lambda x^B.\ve t$.
	  $\AtP{\lambda x^B.\ve t}[A/X]=
	  \lambda x^B.\AtP{\ve t}[A/X]=
	  \lambda x^{B[A/X]}.\AtP{\ve t}[A/X]\overset{IH}{=}
	  \lambda x^{B[A/X]}.\AtP{\ve t[A/X]}=$\\
	  $\AtP{\lambda x^{B[A/X]}.\ve t[A/X]}=
	  \AtP{(\lambda x^B.\ve t)[A/X]}$.

	\item Let $\ve r=\ve t_1\ve t_2$.
	  $\AtP{\ve t_1\ve t_2}[A/X]=
	  \AtP{\ve t_1}[A/X] \AtP{\ve t_2}[A/X] \overset{IH}{=}
	  \AtP{\ve t_1[A/X]} \AtP{\ve t_2[A/X]} =
	  \AtP{\ve t_1[A/X]\ve t_2[A/X]} =
	  \AtP{(\ve t_1\ve t_2)[A/X]}$.

	\item Let $\ve r=\Lambda Y.\ve t$, with $Y\notin FV(A)$. 
	  $\AtP{\Lambda Y.\ve t}[A/X]=
	  \Lambda Y.\AtP{\ve t}[A/X]\overset{IH}{=}
	  \Lambda Y.\AtP{\ve t[A/X]}=
	  \AtP{\Lambda Y.\ve t[A/X]}=
	  \AtP{(\Lambda Y.\ve t)[A/X]}$.

	\item Let $\ve r=\ve t\{B\}$. 
	  $\AtP{\ve t\{B\}}[A/X]=
	  \AtP{\ve t}\{B\}[A/X]=
	  \AtP{\ve t}[A/X]\{B[A/X]\}\overset{IH}{=}
	  \AtP{\ve t[A/X]}\{B[A/X]\}=$\\
	  $\AtP{\ve t[A/X]\{B[A/X]\}}=
	  \AtP{(\ve t\{B\})[A/X]}$.

	\item Let $\ve r=\sui{n}\frac{n_i}{d_i}.\ve r_i$.
	  $\AtP{\sui{n}\frac{n_i}{d_i}.\ve r_i}[A/X]=
	  \pi_A\left( \sui{n}m_i.\AtP{\ve r_i} \right)[A/X]=
	  \pi_A\left( \sui{n}m_i.\AtP{\ve r_i}[A/X]\right)\overset{IH}{=}$\\
	  $\pi_A\left(\sui{n}m_i\AtP{\ve r_i[A/X]}\right)=
	  \AtP{\sui{n}\frac{n_i}{d_i}.\ve r_i[A/X]}=
	  \AtP{(\sui{n}\frac{n_i}{d_i}.\ve r_i)[A/X]}$.
      \end{itemize}
    \item We proceed by induction on $\ve r$.
      \begin{itemize}
	\item Let $\ve r=x^A$.
	  $\AtP{x^A}[\AtP{\ve s}/x]=
	  x^A[\AtP{\ve s}/x]=
	  \AtP{\ve s}=
	  \AtP{x^A[\ve s/x]}$.

	\item Let $\ve r=y^A$,
	  $\AtP{y^A}[\AtP{\ve s}/x]=
	  y^A[\AtP{\ve s}/x]=
	  y^A=
	  \AtP{y^A}=
	  \AtP{y^A[\ve s/x]}$.

	\item Let $\ve r=\lambda y^B.\ve t$.
	  $\AtP{\lambda y^B.\ve t}[\AtP{\ve s}/x]=
	  \lambda y^B.\AtP{\ve t}[\AtP{\ve s}/x]\overset{IH}{=}$

	  \hfill$
	  \lambda y^B.\AtP{\ve t[\ve s/x]}=
	  \AtP{\lambda y^B.\ve t[\ve s/x]}=
	  \AtP{(\lambda y^B.\ve t)[\ve s/x]}$.

	\item Let $\ve r=\ve t_1\ve t_2$.
	  $\AtP{\ve t_1\ve t_2}[\AtP{\ve s}/x]=
	  \AtP{\ve t_1}[\AtP{\ve s}/x] \AtP{\ve t_2}[\AtP{\ve s}/x] \overset{IH}{=}
	  \AtP{\ve t_1[\ve s/x]} \AtP{\ve t_2[\ve s/x]} =
	  \AtP{\ve t_1[\ve s/x]\ve t_2[\ve s/x]} =$\\
	  $\AtP{(\ve t_1\ve t_2)[\ve s/x]}$.

	\item Let $\ve r=\Lambda X.\ve t$, 
	  $\AtP{\Lambda X.\ve t}[\AtP{\ve s}/x]=
	  \Lambda X.\AtP{\ve t}[\AtP{\ve s}/x]\overset{IH}{=}
	  \Lambda X.\AtP{\ve t[\ve s/x]}=
	  \AtP{\Lambda X.\ve t[\ve s/x]}=
	  \AtP{(\Lambda X.\ve t)[\ve s/x]}$.

	\item Let $\ve r=\ve t\{B\}$. Let $FV(\ve s)=\vec{X}$ and $\vec{Y}$ be a set of free variables such that $\ve s[\vec{Y}/\vec{X}][\vec{X}/\vec{Y}]=\ve s$. Then,
	  $\AtP{\ve t\{B\}}[\AtP{\ve s}/x]=
	  \AtP{\ve t}\{B\}[\AtP{\ve s}/x]=
	  \AtP{\ve t}[\AtP{\ve s[\vec{Y}/\vec{X}]/x}]\{B\}[\vec{X}/\vec{Y}]\overset{IH}{=}
	  \AtP{\ve t[\ve s[\vec{Y}/\vec{X}]/x]}\{B\}[\vec{X}/\vec{Y}]$

	  \hfill$
	  =
	  \AtP{\ve t[\ve s[\vec{Y}/\vec{X}]/x]\{B\}}[\vec{X}/\vec{Y}]\overset{item~\ref{it:AtPSubsType}}{=}
	  \AtP{\ve t[\ve s[\vec{Y}/\vec{X}]/x]\{B\}[\vec{X}/\vec{Y}]}=
	  \AtP{(\ve t\{B\})[\ve s/x]}$.

	\item Let $\ve r=\sui{n}\frac{n_i}{d_i}.\ve r_i$.
	  $\AtP{\sui{n}\frac{n_i}{d_i}.\ve r_i}[A/X]=
	  \pi_A\left( \sui{n}m_i.\AtP{\ve r_i} \right)[A/X]=
	  \pi_A\left( \sui{n}m_i.\AtP{\ve r_i}[A/X]\right)\overset{IH}{=}$\\
	  $\pi_A\left(\sui{n}m_i\AtP{\ve r_i[A/X]}\right)=
	  \AtP{\sui{n}\frac{n_i}{d_i}.\ve r_i[A/X]}=
	  \AtP{(\sui{n}\frac{n_i}{d_i}.\ve r_i)[A/X]}$. \qedhere
      \end{itemize}
  \end{enumerate}
\end{proof}

\begin{theorem}\label{thm:AlgtoLp}
  If $\ve r\to^*\sui{n}p_i.\ve t_i$, with $\ve t_i$ in \Alg, with $\sui{n}p_i=1$
  and $\AtP{\ve t_i}\to^*\ve s_i$, then
  $\AtP{\ve r}\to^*\ve s_i$ with probability 
  $p_i\left(\suj{n}p_j\right)^{-1}$ in \Lpplus.
\end{theorem}
\begin{proof}
  Let $\ve r:A$ in \Alg.
  For $i=1,\dots,n$, assume $p_i=\dfrac{n_i}{d_i}$ with $n_i,d_i\in\mathbb{N}^*$.
  We proceed by a case analysis on the last reduction step to reach
  $\sui{n}p_i.\ve t_i$.
  \begin{itemize}
    \item If $\ve r=\sui{n}p_i.\ve t_i$, then
      $\pi_A(\sui{n}(\prod_{\arriba{k=1}{k\neq i}}^n d_kn_i).\AtP{\ve t_i})\to^*
      \pi_A(\sui{n}(\prod_{\arriba{k=1}{k\neq i}}^n d_kn_i).\ve s_i')$
      By Theorem~\ref{thm:toprob}, this term reduces in one step to $\ve s_i'$ with probability
      $$\dfrac{\prod_{\arriba{k=1}{k\neq i}}^nd_kn_i}{\sui{n}\left(\prod_{\arriba{k=1}{k\neq i}}^nd_kn_i\right)}
      =
      \left(\dfrac{\dfrac{n_i}{d_i}}{\sui{n}\dfrac{n_i}{d_i}}\right)
      .
      \left(\dfrac{\prod_{k=1}^nd_k}{\prod_{k=1}^nd_k}\right)
      =
      p_i\left(\suj{n}p_j\right)^{-1}\enspace.
      $$

    \item Consider
      $1.\ve r\eq\ve r$, with $\ve r=\sui{n}p_i.\ve t_i$. We have,
      $\AtP{1.\ve r}=\pi_A(1.\AtP{\ve r})\to^*\pi_A(1.\ve s)$, which reduces with probability one to $\ve s$. Notice that $\ve s$ is a reduct of $\AtP{\sui{n}p_i.\ve t_i}=\pi_A(\sui{n}m_i.\AtP{\ve t_1})$. We conclude with Theorem~\ref{thm:toprob}.

    \item Consider
      $\left(\sum_{i=m+1}^n p_i.\ve t_i\right)
      +\left(\sui{m}p_i.\ve t_i\right)
      \eq
      \sui{n}p_i.\ve t_i$,
      with $1\leq m<n$.
      Since $\ve r:A$, then each $\ve t_i:A$. We have,
      $$\AtP{\sum_{i=m+1}^n p_i.\ve t_i+\sui{m}p_i.\ve t_i}=
      \pi_A\left(
      \sum_{i=m+1}^n
      m_i.
      \AtP{\ve t_i}
      +
      \sui{m}
      m_i.
      \AtP{\ve t_i}
      \right)
      \eq
      \pi_A\left(
      \sui{n}
      m_i.
      \AtP{\ve t_i}
      \right)\enspace.$$
      where $m_i=\prod_{\arriba{k=1}{k\neq i}}^n d_kn_i$.
      We conclude with Theorem~\ref{thm:toprob}.

    \item Consider
      $\lambda x^A.(\ve r+\ve s)\eq\lambda x^A.\ve r+\lambda x^A.\ve s$.
      We have
      $\AtP{\lambda x^A.(\ve r+\ve s)}
      =
      \lambda x^A.(\AtP{\ve r+\ve s})
      =
      \lambda x^A.\pi_A(\AtP{\ve r}+\AtP{\ve s})
      \to^*
      \lambda x^A.\pi_A(\ve r'+\ve s')$
      By Theorem~\ref{thm:toprob}, 
      $\lambda x^A.\pi_A(\ve r'+\ve s')$
      reduces to
      $\lambda x^A.\ve r'$ (which is a reduct of $\AtP{\lambda x^A.\ve r}
      =\lambda x^A.\AtP{\ve r}$),
      with probability $\frac{1}{2}$,
      and to
      $\lambda x^A.\ve s'$ (which is a reduct of $\AtP{\lambda x^A.\ve s}
      =\lambda x^A.\AtP{\ve s}$),
      with probability $\frac{1}{2}$.

    \item Consider $(\lambda x^A.\ve r)~\ve s\re\ve r[\ve s/x]$, with
      $\ve r[\ve s/x]=\sui{n}p_i.\ve t_i$. Then
      $\AtP{(\lambda x^A.\ve r)~\ve s}
      =
      (\lambda x^A.\AtP{\ve r})~\AtP{\ve s}
      \re
      \AtP{\ve r}[\AtP{\ve s}/x]$ which, by Lemma~\ref{lem:AtPSubstitution}(\ref{it:AtPSubsTerm}),
      is equal to
      $\AtP{\ve r[\ve s/x]}=
      \AtP{\sui{n}p_i.\ve t_i}$ and this, by definition is equal to $\pi_A\left( \sui{n}(\prod_{\arriba{k=1}{k\neq i}}^n d_kn_i).\AtP{\ve t_i} \right)$. We conclude with Theorem~\ref{thm:toprob}.

    \item Consider $(\Lambda X.\ve r)\{A\}\re\ve r[A/X]$, with
      $\ve r[A/X]=\sui{n}p_i.\ve t_i$.
      Then,
      $\AtP{\Lambda X.\ve r\{A\}}=\Lambda X.\AtP{\ve r}\{A\}\re\AtP{\ve r}[A/X]$, which by Lemma~\ref{lem:AtPSubstitution}(\ref{it:AtPSubsType}) is equal to 
      $\AtP{\ve r[A/X]}=
      \AtP{\sui{n}p_i.\ve t_i}=
      \pi_A\left(\! \sui{n}(\prod_{\arriba{k=1}{k\neq i}}^n d_kn_i).\AtP{\ve t_i}\! \right)$. 
      We conclude with Theorem~\ref{thm:toprob}.

    \item Consider $p.q.\ve r\re pq.\ve r$. Let $p.q.\ve r:A$. Since $pq.\ve r=\sui{n}p_i.\ve t_i$ with $\sui{n}p_i=1$, we have $n=1$ and $pq=p_1=1$. Also, since $p.q.\ve r$ is a term, $p=q=1$. So, we have
      $\AtP{1.1.\ve r}=
      \pi_A(1.\AtP{1.\ve r})=
      \pi_A(1.\pi_A(1.\AtP{\ve r}))\to^*
      \pi_A(1.\pi_A(1.\ve s))$
      Notice that this term reduces with probability $1$ to
      ${\pi_A(1.\ve s)}$, which is a reduct of
      ${\pi_A(1.\AtP{\ve r})}=\AtP{1.\ve r}$.

    \item Consider $p.(\ve r_1+\ve r_2)\re p.\ve r_1+p.\ve r_2$. Since $p.\ve r_1+p.\ve r_2=\sui{n}p_i.\ve t_i$, with $\sui{n}p_i=1$, we have $n=2$ and $p=\frac{1}{2}$, however in such case $\frac{1}{2}.(\ve r_1+\ve r_2)$ is a pseudo-term, not a term.

    \item Consider $p.\ve r+q.\ve r\re (p+q).\ve r$. Since $(p+q).\ve r=\sui{n}p_i.\ve t_i$, with $\sui{n}p_i=1$, we have $n=1$ and $p+q=1$. Let $p=\frac{m}{d}$, then $q=\frac{d-m}{d}$. So, $\AtP{p.\ve r+q.\ve r}=\pi_A(dm.\AtP{\ve r}+(d(d-m)).\AtP{\ve r})\re\pi_A(d^2.\AtP{\ve r})$, which reduces with probability $1$ to ${\ve s}$, where $\AtP{\ve r}\to^*{\ve s}$.
    \item Contextual rules are straightfoward.\qedhere
  \end{itemize}
\end{proof}

\subsection{Back from \texorpdfstring{\Lpplus}{L+p}~to~\texorpdfstring{\Alg}{Alg}}
The inverse translation is given by
\[
  \begin{array}{r@{~=~}l@{\qquad}r@{~=~}l@{\qquad}r@{~=~}l}
    \PtA{x^A} & x^A
    &
    \PtA{\ve r\ve s} & \PtA{\ve r}\PtA{\ve s}
    &
    \PtA{\ve r\{A\}} & \PtA{\ve r}\{A\}
    \\
    \PtA{\lambda x^A.\ve r} & \lambda x^A.\PtA{\ve r}
    &
    \PtA{\Lambda X.\ve r} & \Lambda X.\PtA{\ve r}
    &
    \PtA{\ve r+\ve s} & \PtA{\ve r}+\PtA{\ve s}
    \\
    \multicolumn{4}{r}{\mbox{If }\pi_A(\ve t)\re\ve s_i\mbox{ with probability }p_i,\mbox{ for }i=1,\dots,n,}
    &
    \PtA{\pi_A(\ve t)} & \sui{n}p_i.\PtA{\ve s_i}
  \end{array}
\]

\begin{remark} This translation does not admit translating a term of the form $\pi_A(\ve t)$ in normal form. Moreover, let $\Pi$ be the rule
  ``$\pi_{A\Rightarrow B}(\ve r)\ve s\eq\pi_B(\ve r\ve s)\mbox{ with }\ve r:A\Rightarrow(B\wedge C)$'',
  then the translation keep reductions, except for the one using rule $\Pi$, as expressed in Theorem \ref{thm:PtA}.
\end{remark}

\begin{lemma}\label{lem:PtASubstitution}~
  \begin{multicols}{2}
  \begin{enumerate}
    \item\label{it:PtASubsType} $\PtA{\ve r}[A/X]=\PtA{\ve r[A/X]}$
    \item\label{it:PtASubsTerm} $\PtA{\ve r}[\PtA{\ve s}/x]=\PtA{\ve r[\ve s/x]}$
  \end{enumerate}
  \end{multicols}
\end{lemma}
\begin{proof}
  Both items follow by induction on $\ve r$. Cases $x^B$, $\lambda x^B.\ve t$, $\ve t_1\ve t_2$, $\Lambda Y.\ve t$ and $\ve t\{B\}$ are analogous to those in proof of Lemma~\ref{lem:AtPSubstitution}. Hence we only need to verify the case $\pi_B(\ve t)$, when $\ve r\re\ve r_i$ with probability $p_i$, for $i=1,\dots,n$.
  \begin{enumerate}
    \item $\PtA{\pi_B(\ve t)}[A/X]=
      (\sui{n}p_i.\PtA{\ve r_i})[A/X]=
      \sui{n}p_i.\PtA{\ve r_i}[A/X]$,
      which by the induction hypothesis, is equal to
      $\sui{n}p_i.\PtA{\ve r_i[A/X]}=
      \PtA{\pi_{B[A/X]}(\ve t[A/X])}=
      \PtA{(\pi_B(\ve t))[A/X]}$.
    \item $\PtA{\pi_B(\ve t)}[\ve s/x]=
      (\sui{n}p_i.\PtA{\ve r_i})[\ve s/x]=
      \sui{n}p_i.\PtA{\ve r_i}[\ve s/x]$,
      which by the induction hypothesis, is equal to
      $\sui{n}p_i.\PtA{\ve r_i[\ve s/x]}=
      \PtA{\pi_B(\ve t[\ve s/x])}=
      \PtA{(\pi_B(\ve t))[\ve s/x]}$.\qedhere
  \end{enumerate}
\end{proof}

\begin{theorem}
  \label{thm:PtA} Let $\ve r,\ve s,\ve s_i$ in \Lpplus. 
  \begin{itemize}
    \item If $\ve r\eq\ve s$, then $\PtA{\ve r}\eq\PtA{\ve s}$.
    \item If $\ve r\re\ve s$, with probability $1$, then $\PtA{\ve r}\re\PtA{\ve s}$, except if the reduction is done by rule $\Pi$.
    \item If $\ve r\re\ve s_i$ with probability $p_i$, for $i=1,\dots,n$, then 
      $\PtA{\ve r}=\sui{n}p_i.\PtA{\ve s_i}$.
  \end{itemize}
\end{theorem}
\begin{proof}
  Case by case analysis. 
  \begin{itemize}
    \item Consider $\ve r+\ve s\eq\ve s+\ve r$. Notice that 
      $\PtA{\ve r+\ve s}=
      \PtA{\ve r}+\PtA{\ve s}\eq\PtA{\ve s}+\PtA{\ve r}=
      \PtA{\ve s+\ve r}$.
    \item Consider $(\ve r+\ve s)+\ve t\eq\ve r+(\ve s+\ve t)$. Notice that
      $\PtA{(\ve r+\ve s)+\ve t}=
      (\PtA{\ve r}+\PtA{\ve s})+\PtA{\ve t}\eq
      \PtA{\ve r}+(\PtA{\ve s}+\PtA{\ve t})=
      \PtA{\ve r+(\ve s+\ve t)}$.
    \item Consider $(\ve r+\ve s)\ve t\eq\ve r\ve t+\ve s\ve t$. Notice that
      $\PtA{(\ve r+\ve s)\ve t}=
      (\PtA{\ve r}+\PtA{\ve s})\PtA{\ve t}\eq
      \PtA{\ve r}\PtA{\ve t}+\PtA{\ve s}\PtA{\ve t}=
      \PtA{\ve r\ve t+\ve s\ve t}$.
    \item Consider $\lambda x^A.(\ve r+\ve s)\eq\lambda x^A.\ve r+\lambda x^A.\ve s$. Notice that
      $\PtA{\lambda x^A.(\ve r+\ve s)}=
      \lambda x^A.(\PtA{\ve r}+\PtA{\ve s})\eq
      \lambda x^A.\PtA{\ve r}+\lambda x^A.\PtA{\ve s}=
      \PtA{\lambda x^A.\ve r+\lambda x^A.\ve s}$.
    \item Consider $(\lambda x^A.\ve r)\ve s\re\ve r[\ve s/x]$. Notice that
      $\PtA{(\lambda x^A.\ve r)\ve s}=
      (\lambda x^A.\PtA{\ve r})\PtA{\ve s}\re
      \PtA{\ve r}[\PtA{\ve s}/x]$, and this, by Lemma~\ref{lem:PtASubstitution}(\ref{it:PtASubsTerm}), is equal to
      $\PtA{\ve r[\ve s/x]}$.
    \item Consider $(\Lambda X.\ve r)\{A\}\re\ve r[A/X]$. Notice that
      $\PtA{(\Lambda X.\ve r)\{A\}}=
      \Lambda X.\PtA{\ve r}[A/X]\re
      \PtA{\ve r}[A/X]$, and this, by Lemma~\ref{lem:PtASubstitution}(\ref{it:PtASubsType}), is equal to
      $\PtA{\ve r[A/X]}$.
    \item Consider $\pi_A(\sui{n}m_i.\ve r_i+\ve s)\re\ve r_i$ with probability $\frac{m_i}{\suj{n}m_j}$, where $\ve r_i:A$ and $\ve s\notder A$ are closed normal terms. Notice that, by definition,
      $\PtA{\pi_A(\sui{n}m_i.\ve r_i+\ve s)}=
      \sui{n}\frac{m_i}{\suj{n}m_j}.\PtA{\ve r_i}$.\qedhere
  \end{itemize}
\end{proof}

\section{Conclusion}
In this paper we have defined a probability space on the execution traces of non-confluent abstract rewrite systems. We define a sample space on strategies deciding the rewrite to apply at each state (cf.~Definition~\ref{def:strategy}). 

Our main motivation has been to be able to use this probability space in non-deterministic calculi, hence being able to encode a probability superposition of the kind $\alpha.\ve t+\beta.\ve r$, with $\alpha+\beta=1$, as a term having probability $\alpha$ of rewriting to $\ve t$ and probability $\beta$ of rewriting to $\ve r$. 
As an example, we provided such an encoding from an algebraic calculus into a non-deterministic calculus.

\paragraph{Acknowledgements}
We would like to thank Laurent Duvernet for enlightening discussions.
This work has been funded by the ANR-10-JCJC-0208 CausaQ grant.

\bibliography{biblio}

\end{document}